\documentclass[a4paper,11pt]{article}
\usepackage[sc]{mathpazo}
\usepackage{microtype}
\usepackage[utf8]{inputenc}
\usepackage{graphicx,paralist}
\usepackage{amsmath, amssymb, amsthm}
\usepackage{comment,hyperref}
\usepackage{thmtools}
\usepackage{tikz}
\usepackage{pgfplots}
\usepackage{varwidth}
\pgfplotsset{compat=1.10}
\usepgfplotslibrary{fillbetween}
\usetikzlibrary{backgrounds}
\usetikzlibrary{patterns}
\usepackage{enumitem}
\usepackage{geometry}
\usepackage{array}
\usepackage{multirow}
\usepackage{stmaryrd}

\newcommand{\calT}{\mathcal{T}}
\newcommand{\calY}{\mathcal{Y}}
\newcommand{\calJ}{\mathcal{J}}
\newcommand{\calE}{\mathcal{E}}

\newcommand{\QQ}{\mathsf{Q}}

\newcommand{\EE}{\mathbb{E}}
\newcommand{\calR}{\mathcal{R}}
\newcommand{\calB}{\mathcal{B}}

\newcommand{\calP}{\mathcal{P}}

\newcommand{\calC}{\mathcal{C}}
\newcommand{\calX}{\mathcal{X}}
\newcommand{\calV}{\mathcal{V}}
\newcommand{\calD}{\mathcal{D}}

\newcommand{\calA}{\mathcal{A}}

\newcommand{\sos}{\mathsf{SoS}}

\newcommand{\calI}{\mathcal{I}}
\newcommand{\opt}{\mathsf{opt}}

\renewcommand{\epsilon}{\varepsilon}

\newcommand{\FPT}{\mathsf{FPT}}

\renewcommand{\EE}{\mathsf{E}}
\newcommand{\PP}{\mathsf{P}}

\newcommand{\depth}{\mathsf{depth}}

\newcommand{\parent}{\mu}

\newcommand{\calS}{\mathcal{S}}

\newcommand{\capac}{\mathsf{cap}}
\newcommand{\dem}{\mathsf{dem}}
\newcommand{\LP}{\mathsf{LP}}
\newcommand{\NP}{\mathsf{NP}}
\renewcommand{\P}{\mathsf{P}}

\newtheorem{theorem}{Theorem}
\newtheorem{lemma}{Lemma}

\newtheorem{definition}{Definition}

\newtheorem{proposition}{Proposition}
\newtheorem{claim}{Claim}

\usepackage{xspace}

\usepackage{multirow}
\usepackage[noend]{algpseudocode}
\usepackage{algorithm,algorithmicx}
\usepackage{xcolor}

\algnewcommand{\LeftComment}[1]{\Statex \(\triangleright\) #1}

\hypersetup{
	colorlinks,
	linkcolor={red!50!black},
	citecolor={blue!50!black},
	urlcolor={blue!80!black}
}
\begin{document}
	\algrenewcommand\algorithmicrequire{\textbf{Input:}}
	\algrenewcommand\algorithmicensure{\textbf{Output:}}
	
	\title{A 2-Approximation for the Bounded Treewidth Sparsest Cut Problem in $\FPT$ Time
	\vspace{1cm}
	 }	
 
	\author{
	Vincent Cohen-Addad			
	\thanks{Google Research Zurich, Switzerland. {\tt vcohenad@gmail.com}}
	\and Tobias M\"{o}mke
	\thanks{University of Augsburg, Department of Computer Science, Germany. {\tt moemke@informatik.uni-augsburg.de}. Partially supported by 
     DFG Grant 439522729 (Heisenberg-Grant) and DFG Grant 439637648 (Sachbeihilfe)}
	\and Victor Verdugo			
	\thanks{Universidad de O'Higgins, Institute of Engineering Sciences, Chile. {\tt victor.verdugo@uoh.cl}}
	}
\date{\vspace{-1em}}
\maketitle
\begin{abstract}
In the non-uniform sparsest cut problem, we are given a supply graph $G$ and a demand graph $D$, both with the same set of nodes $V$. The goal is to find a cut of $V$ that minimizes the ratio of the total capacity on the edges of $G$ crossing the cut over the total demand of the crossing edges of $D$. In this work, we study the non-uniform sparsest cut problem for supply graphs with bounded treewidth $k$. For this case, Gupta, Talwar and Witmer [STOC 2013] obtained a 2-approximation with polynomial running time for fixed $k$, and the question of whether there exists a $c$-approximation algorithm for a constant $c$ independent of $k$, that runs in $\FPT$ time, remained open. We answer this question in the affirmative. We design a 2-approximation algorithm for the non-uniform sparsest cut with bounded treewidth supply graphs that runs in $\FPT$ time, when parameterized by the treewidth. Our algorithm is based on rounding the optimal solution of a linear programming relaxation inspired by the Sherali-Adams hierarchy. In contrast to the classic Sherali-Adams approach, we construct a relaxation driven by a tree decomposition of the supply graph by including a carefully chosen set of lifting variables and constraints to encode information of subsets of nodes with super-constant size, and at the same time we have a sufficiently small linear program that can be solved in $\FPT$ time.
\end{abstract}

\section{Introduction}

In the non-uniform sparsest cut problem, we are given two weighted graphs $G$ and $D$ on the same set of nodes $V$,
such that $G=(V,E_G)$ is the so-called supply graph,
and $D=(V,E_D)$ is the so-called demand graph.
For every edge $e\in E_G$ we have a positive integer weight $\capac(e)$ called capacity, and for every edge $e\in E_D$ we have a positive integer weight
$\dem(e)$ called the demand. 
An instance $\calI$ is given by a tuple $(G,D,\capac,\dem)$ and we denote by $|\calI|$ the encoding length of an instance $\calI$.
The goal is to compute a non-empty subset of nodes $S \subseteq V$ that minimizes
\[\phi(S) = \frac{\sum_{e\in \delta_G(S)}\capac(e)}{\sum_{e\in \delta_D(S)}\dem(e)},\]
where $\delta_G(S)=\{e\in E_G:|e\cap S|=1\}$ and $\delta_D(S)=\{e\in E_D:|e\cap S|=1\}$.
Since this problem is $\NP$-hard \cite{matula1990sparsest}, the focus has been on the design of approximation algorithms. 
In this line of work, Agrawal, Klein, Rao and Ravi \cite{klein1990approximation,klein1995approximate} took the first major step by
describing a $O(\log D\log C)$-approximation algorithm,
where $D$ is the total sum of the demands and $C$ is the total sum of the capacities.
Currently, the best approximation factor is $O(\sqrt{\log n}\log \log n)$ due to Arora, Lee and Naor \cite{arora2008euclidean}. 
The \emph{uniform} version of the problem, where the demand graph is unweighted and  complete, has received a lot of attention
through the years. The best bound for this problem is slightly better: $O(\sqrt{\log n})$ \cite{DBLP:journals/jacm/AroraRV09}.

The non-uniform sparsest cut problem is hard to approximate within a constant factor, for any constant,
under the unique games conjecture \cite{chawla2006hardness,GTW13,khot2015unique}.
Therefore, the problem has also been studied under the assumption that the supply graph belongs to a specific family of graphs.
Most notable examples include planar graphs, graph excluding a fixed minor, and bounded treewidth graphs.
In this paper, we focus on the latter (see Section~\ref{sec:relatedwork} for further related work on minor-closed families).

For inputs to the problem where the supply graph has treewidth at most $k$, Chlamtac, Krauthgamer and Raghavendra \cite{chlamtac2010approximating}
designed a $C(k)$-approximation algorithm that runs in time $2^{O(k)}|\calI|^{O(1)}$, where $C$ is a double exponential function of the treewidth $k$. 
Later, Gupta, Talwar and Witmer designed a $2$-approximation algorithm that runs in time $|\calI|^{O(k)}$ \cite{GTW13}. 
However, these two results are only
complementary: Chlamtac, Krauthgamer and Raghavendra's algorithm is Fixed-Parameter Tractable ($\FPT$) in the treewidth $k$ of the supply graph,
(i.e. $f(k)|\calI|^{O(1)}$ time for some computable function $f$), while the algorithm of Gupta, Talwar and Witmer is not. 
The approximation factor achieved by Gupta, Talwar and Witmer is independent of $k$, and furthermore, they show that there is no $(2-\varepsilon)$-approximation algorithm for any $\varepsilon>0$ on graphs with constant treewidth, assuming the Unique Games Conjecture and that there is no $1/\alpha_{\text{GW}} - \epsilon$ approximation algorithm for treeewidth 2 graphs unless  $\P = \NP$.
This left open the question of whether there exists a 2-approximation algorithm that runs in $\FPT$ time, when parameterized by the treewidth. 
We answer this question in the affirmative and show the following result.
\begin{theorem}
	\label{thm:main-sparsest}
	There is an algorithm that computes a $2$-approximation for every instance $\calI=(G,D,\capac,\dem)$ of the non-uniform sparsest cut problem in time 
	\[2^{2^{O(k)}}|\calI|^{O(1)},\]
	where $k$ is the treewidth of the supply graph $G$.
\end{theorem}
Following the argumentation of Gupta et al.~\cite{GTW13}, for treewidth $k$ graphs our result implies a $2^{2^{O(k)}}|\calI|^{O(1)}$-time $2$-approximation to the minimum-distortion $\ell_1$ embedding.

The results obtained in the predecessor papers~\cite{chlamtac2010approximating,GTW13} where based on rounding certain linear programs obtained through the Sherali-Adams {\it lift \& project} hierarchy \cite{SA90}.
Our approximation algorithm is also based on rounding a linear program with a fractional objective given by the non-uniform sparsest cut value, but we construct this linear program in a different way, with the goal of obtaining a linear program of smaller size, but sufficiently strong in terms of gap.  
If we followed the classic Sherali-Adams approach, the relaxation of {\it level} $\ell$ would be constructed by using a variable encoding the value of any subset of the original variables up to size $\ell$, and it would take $n^{O(\ell)}$ time to solve this relaxation, where $n$ is the number of nodes in the graph.
In particular, solving a relaxation of level $\Theta(k)$ would take $n^{O(k)}$ time, which in principle rules out the possibility of achieving $\FPT$ running time by applying directly this approach.
In order to overcome this problem, we construct a linear programming relaxation driven by a tree decomposition of the supply graph $G$, where the variables are carefully chosen with the goal of encoding information of subsets of nodes with super-constant size, and at the same time the number of variables and constraints is sufficiently small so we can solve the relaxation in $\FPT$ time. 
We show that the relaxation is strong enough to get a 2-approximation by rounding the optimal fractional solution. 
The construction of our relaxation and the analysis of our algorithm can be found in Section \ref{sec:faster-algorithm}.

\subsection{Related Work}
\label{sec:relatedwork}
Despite the difficulties in approximating the non-uniform sparsest cut problem in general graphs, there are several other results for restricted families of graphs.
The case in which $G$ is planar has received a lot of attention.
Quite recently, Cohen-Addad, Gupta, Klein and Li \cite{cohen2021quasipolynomial} showed the existence of a quasi-polynomial time $(2+\varepsilon)$-approximation for the non-uniform sparsest cut problem in the planar case.
To get this result they combine a patching lemma approach with linear programming techniques.
We remark that for the planar case there is no polynomial time $1/(0.878+\varepsilon)\approx (1.139-\varepsilon)$-approximation algorithm under the unique games conjecture \cite{GTW13}.

Other families with constant factor approximation algorithms are outerplanar graphs \cite{okamura1981multicommodity}, series-parallel \cite{gupta2004cuts,chekuri2013flow,lee2010coarse}, $k$-outerplanar graphs \cite{chawla2006hardness}, graphs obtained by $2$-sums of $K_4$ \cite{chakrabarti2008embeddings} and graphs with constant pathwidth \cite{lee2013pathwidth}.
The impact of the treewidth parameter has also been studied in the context of polynomial optimization \cite{bienstock2018lp}. 
Finally, we mention that the Sherali-Adams hierarchy has been useful to design algorithms in other minor-free and bounded treewidth graph problems, including independent set and vertex cover~\cite{magen2009,bienstock2004tree}, and also
in several recent results on scheduling and clustering~\cite{VVW20,Garg18,maiti2020scheduling,davies2021scheduling,aprile2020tight}. 

Independent of our work, Chalermsook et al.~\cite{CKMSUV21} obtained a $O(k^2)$-approximation algorithm for sparsest cut in treewidth $k$ graphs, with running time $2^{O(k)} \cdot \text{poly}(n)$ and, for arbitrary $\epsilon > 0$, an $O(1/\epsilon^2)$-approximation algorithm with running time $2^{O(k^{1+\epsilon}/\epsilon)} \cdot \text{poly}(n)$.
Observe that these results are incomparable with our result: they obtain an asymptotically lower running time, whereas the obtained (constant) approximation ratio is considerably larger than 2.
Similar to our result, they build on the techniques from \cite{chlamtac2010approximating,GTW13}. However, their approach is based on a new measure for tree decompositions which they call the combinatorial diameter.

\section{Preliminaries: Tree Decompositions}

A tree decomposition of a graph $G=(V,E)$ is a pair $(\mathcal{X},\calT)$ where $\calT=(\calX,E_{\calT})$ is a tree and $\calX$ is a collection of subsets of nodes in $V$ called {\it bags}.
Each bag is a node in the tree $\calT$.
Furthermore, the pair $(\mathcal{X},\calT)$ satisfies the following conditions.
\begin{enumerate}[label=(\arabic*)]
	\item Every node in $V$ is in at least one bag, that is, $\displaystyle \cup_{X\in \calX}X=V$. \label{tree-1}
	\item For every edge $\{u,v\}\in E$ there exists a bag $X\in \calX$ such that $\{u,v\}\subseteq X$. \label{tree-2}
	\item For every node $u\in V$ the bags containing $u$ induce a subtree of $\calT$. \label{tree-3}
\end{enumerate}
The {\it width} of the tree decomposition $(\calX,\calT)$ corresponds to
the size of the largest bag in the tree decomposition, minus one.
The treewidth of $G$ is the minimum possible width of a tree decomposition for $G$.
We typically consider the tree $\calT$ to be rooted, and we denote its root by $\calR$.
We denote by $\depth(\calT)$ the depth of the tree $\calT$ and we say that a bag $X$ is at level $\ell$ if the distance from the root $\calR$ to $X$ in the tree $\calT$ is equal to $\ell$.
We denote by $\parent(X)$ the parent of $X$ in the tree $\calT$.
The intersection between a non-root bag $X$ and the parent bag, $\mu(X)\cap X$, is the called the {\it adhesion} of the bag $X$.
We say that a bag $Y$ is a {\it descendant} of $X$ if $X\ne Y$ and the bag $X$ belongs to the unique path in $\calT$ from $Y$ to the root, and in this case we say that $X$ is an {\it ancestor} of $Y$.

The following result due to Bodlaender \cite[Theorem 4.2]{bodlaender1988nc} states the existence of tree decompositions with a particular structure that is useful for our algorithm.

\begin{lemma}[\cite{bodlaender1988nc}]
	\label{lem:balanced-decomposition}
	Let $G$ be a graph with $n$ nodes and treewidth $k$.
	Then, there exists a tree decomposition $(\calX,\calT)$ of $G$ satisfying the following:
	\begin{enumerate}[label=(\alph*)]
		\item $\calT$ is a binary tree and $\depth(\calT)\in O(\log n)$.
		\item For every $X\in \calX$ we have that $|X|\le 3k+3$.
	\end{enumerate}
	The tree decomposition $(\calX,\calT)$ can be computed in time $2^{O(k^3)}n$.
\end{lemma}

\section{The LP Relaxation and the Rounding Algorithm}
\label{sec:faster-algorithm}

Our algorithm is based on rounding the optimal solution of a linear programming relaxation for the non-uniform sparsest cut problem.
In Section \ref{sec:sparsest-relaxation} we provide the construction of our linear programming relaxation and in Section \ref{sec:sparsest-analysis} we provide the rounding algorithm and the proof of Theorem \ref{thm:main-sparsest}.
In the following lemma we show the existence of a tree structure that we use to construct the linear program (see also Fig.~\ref{fig:bags}).
\begin{lemma}
\label{lem:super-node}
Let $G$ be a graph with treewidth $k$ and let $\ell$ be a positive integer.
Then, there exists a tree decomposition $(\calY,\calE)$ of $G$ such that the following holds:
\begin{enumerate}[label=(\alph*)]
	\item The width of $(\calY,\calE)$ is  $O(2^{\ell}k)$ and $\depth(\calE)\in O(\log(n)/\ell)$.\label{supernode1}
	\item For every non-root bag $Y\in \calY$, the size of the adhesion of $Y$ is $O(k)$. \label{supernode2}
\end{enumerate}	
The decomposition $(\calY,\calE)$ can be found in $2^{O(k^3)}n$ time.
\end{lemma}

\begin{proof}	
Consider a tree decomposition $(\calX,\calT)$ satisfying the conditions guaranteed by Lemma \ref{lem:balanced-decomposition}.
That is, the tree $\calT$ is binary, $\depth(\calT)\in O(\log(n))$ and $|X|\le 3k+3$ for every bag $X\in \calX$.
Let $\calX_{\ell}$ be the set of all bags in $\calX$ at level $i(\ell-1)$ with $i\in \{0,1,\ldots,\lfloor\depth(\calT)/(\ell-1)\rfloor\}$.
For every bag $X\in \calX_{\ell}$ with level $i(\ell-1)$ consider the set $Y_X$ given by the union of $X$ with all its descendant bags at level $i(\ell-1)+j$ with $j\in \{0,1,\ldots,\ell-1\}$.
Since $\calT$ is binary we have that $|Y_X|\le 2^{\ell}(3k+3)= O(2^{\ell}k)$ and we define $\calY=\{Y_X:X\in \calX_{\ell}\}$.
We define the edges of the tree $\calE$ with nodes $\calY$ as follows:
For every bag $X$ at level $i(\ell-1)$ with $i\in \{0,1,\ldots,\lfloor\depth(\calT)/(\ell-1)\rfloor-1\}$, we have the edge $\{Y_X,Y_{X'}\}$ in the tree if $X’$ is a descendant of $X$ and the level of $X’$ is $(i+1)(\ell-1)$ (See Figure \ref{fig:bags}).
In particular, the depth of $\calE$ is at most $O(\depth(\calT)/\ell)=O(\log(n)/\ell)$.
This proves condition \ref{supernode1}.

Now let $\{Y_X,Y_{X'}\}$ be any edge of the tree $\calE$ and consider a node $u\in Y_X\cap Y_{X'}$.
Since $(\calX,\calT)$ is a tree decomposition, we have that the set of bags in $\calX$ that contain the node $u$ induces a subtree of $\calT$.
Furthermore, $X'$ is the unique bag of $\calX$ in the intersection of the set of bags defining $Y_X$ and $Y_{X'}$, and therefore the connectivity of this subtree implies that $u\in X’$.
Since $|X'|\le 3k+3$, we conclude that the size of the adhesion of any bag $Y\in \calY$ is $O(k)$.
This shows that condition \ref{supernode2} holds.
The fact that $(\calY,\calE)$ is a tree decomposition holds since the construction preserves conditions \ref{tree-1}-\ref{tree-3} from $(\calX,\calT)$.
\end{proof} 
\begin{figure}
    \begin{center}
        \includegraphics[scale=1.5]{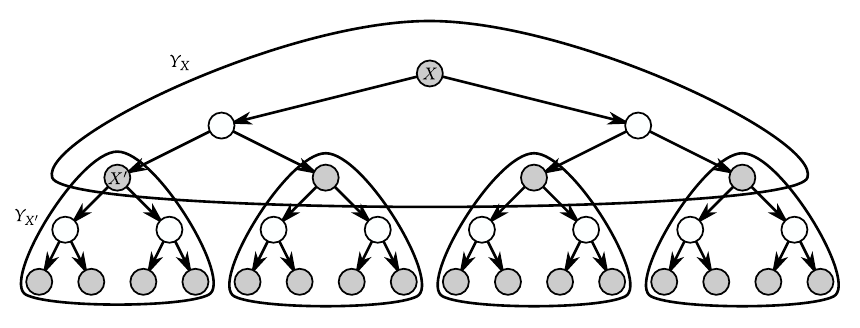}
        \caption{\label{fig:bags}Tree decomposition from Lemma~\ref{lem:super-node} for $\ell = 3$. The gray nodes form the set $\calX_\ell$.}
    \end{center}
\end{figure}
\begin{definition}
Given a graph $G$, we say that a tree decomposition $\Theta=(\calY,\calE)$ satisfying properties \ref{supernode1}-\ref{supernode2} is a {\it $(k,\ell)$-decomposition} of $G$.
\end{definition}
Given a bag $Y\in \calY$, we denote by $\calP_{\Theta}^Y$ the subset of bags that belong to the path from $Y$ to the root $\calR$ in the tree $\calE$.
We denote by $\calJ_Y$ the adhesion of $Y$. 
Furthermore, let 
\[\calV_{\Theta}^{Y}=\bigcup_{Z\in \calP_{\Theta}^Y}\calJ_Z,\] 
and for every pair of non-root bags $Y,Z\in \calY$ let $\calS_{\Theta}(Y,Z)$ be the power set of $(Y\cup \calV_{\Theta}^Y)\cup (Z\cup\calV_{\Theta}^Z).$
Finally, let 
\[\calS_{\Theta}=\bigcup_{Y,Z\in \calY}\calS_{\Theta}(Y,Z).\]
Observe that for every bag $Y\in \calY$, the size of $Y\cup \calV_{\Theta}^Y$ is $ O(2^{\ell}k+k\log(n)/\ell)$.

\subsection{The LP Relaxation}
\label{sec:sparsest-relaxation}
Consider a positive integer $\ell$ and an instance $(G,D,\capac,\dem)$ where $G$ has treewidth $k$.
Let $\Theta=(\calY,\calE)$ be a $(k,\ell)$-decomposition of the supply graph $G$.
In what follows we describe our LP relaxation, inspired by the Sherali-Adams hierarchy \cite{SA90} and the predecessor works~\cite{chlamtac2010approximating,GTW13}.
In this linear program there are two types of variables. 
The variable $x(S,T)$, with $S\in \calS_{\Theta}$ and $T\subseteq S$, indicates that the cut solution $C$ satisfies that $C\cap S=T$. 
The variable $y(\{u,v\})$ for $u,v\in V$ with $u\ne v$, indicates whether the nodes $u$ and $v$ fall in different sides of the cut.
For notation simplicity, we sometimes denote the union between a set $A$ and a singleton $\{a\}$ by $A+a$.
Consider the following linear fractional program:

\begin{align}
\text{minimize} \quad\quad \frac{\sum_{e\in E_G}  \capac(e)y(e)}{\sum_{e\in E_D}\dem(e)y(e)} \quad \quad \quad& \label{eq:objective}\\
\text{subject to}  \quad x(\{u,v\},u)+x(\{u,v\},v))& = y(\{u,v\}) \quad \text{ for every }u,v\in V \text{ with } u\ne v,\label{eq:cut-indicating}\\
\sum_{A\subseteq S} x(S,A)& = 1 \quad \quad\;\;\quad \quad \text{ for every }S\in \calS_{\Theta}, \label{eq:prob-measure}\\
x(S,A)& \ge 0 \quad \quad\;\;\quad \quad \text{ for every }S\in \calS_{\Theta} \text{ and }A\subseteq S, \label{eq:prob-positive}\\
x(S+u,A)+x(S+u,A+u)& = x(S,A) \quad \begin{array}{@{}c@{}}  \;\;\;\text{ for every }S\subseteq V, u\notin S\text{ such that }\\S+u\in \calS_{\Theta} \text{ and }A\subseteq S.\end{array} \label{eq:prob-consistent}
\end{align}

The feasible region of this linear program is a polytope encoding the cuts in $V$.
Indeed, given any cut $C$, define $U_j=1$ if $j\in C$ and zero otherwise.
For every $S\in \calS_{\Theta}$ and $A\subseteq S$, define $x(S,A)=\prod_{j\in A}U_j\prod_{j\in S\setminus A}(1-U_j)$ and $y(\{u,v\})=U_u(1-U_v)+U_v(1-U_u)$.
The solution $(x,y)$ satisfies conditions \eqref{eq:cut-indicating}-\eqref{eq:prob-consistent}.
We remark that \eqref{eq:prob-consistent} is valid for every cut since given a subset $S$ and a node $u\notin S$, the intersection between $S+u$ and a cut $C$ is either $C\cap S$ or $(C\cap S)+u$, which are the two possibilities in the left hand side of \eqref{eq:prob-consistent}.
Since for every bag $Y\in \calY$ the size of $Y\cup \calV_{\Theta}^Y$ is $O(2^{\ell}k+k\log(n)/\ell)$, we get
\begin{align*}
	|\calS_{\Theta}(Y,Z)|&= 2^{O(k(2^{\ell}+\log(n)/\ell))}\text{ for any pair of bags }Y,Z\in \calY,\\
	|\calS_{\Theta}|&\le \sum_{Y,Z\in \calY}|\calS_{\Theta}(Y,Z)|=  n^2 2^{O(k(2^{\ell}+\log(n)/\ell))},
\end{align*}
and therefore the number of variables and constraints in the linear fractional program is 
\[O\Big(|\calS_{\Theta}|\cdot 2^{\max\{|S|:S\in \calS_{\Theta}\}}\Big)= n^2 2^{O(k(2^{\ell}+\log(n)/\ell))}.\]
By using a standard equivalent reformulation the linear fractional program \eqref{eq:objective}-\eqref{eq:prob-consistent} can be solved by a linear program with one additional variable and constraint \cite{boyd2004convex}. 
\subsection{The Rounding Algorithm}
\label{sec:sparsest-analysis}
In this section we describe our algorithm for the non-uniform sparsest cut problem.
Before stating the algorithm, we introduce an object that will be used in the analysis.
Recall that $G$ is of treewidth $k$ and $\Theta=(\calY,\calE)$ is a $(k,\ell)$-decomposition of $G$. 
\begin{definition}
	\label{def:prob-functions}
Given a feasible solution $(x,y)$ satisfying \eqref{eq:cut-indicating}-\eqref{eq:prob-consistent}, we define the function given by $f_{x,\Theta}^{\calR}(A)=x(\calR,A)$ for every $A\subseteq \calR$, where $\calR$ is the root bag of $\calE$.
Furthermore, given any non-root bag $Y\in \calY$ and a subset $T\subseteq \calV_{\Theta}^Y$ such that $x(\calV_{\Theta}^Y,T)>0$, we define the function given by 
\begin{equation}f_{x,\Theta}^{T,Y}(A)=\frac{x(\calV_{\Theta}^Y\cup Y,T\cup A)}{x(\calV_{\Theta}^Y,T)} \label{eq:conditional-prob}
\end{equation}
for every $A\subseteq Y\setminus \mu(Y)$, where $\mu(Y)$ is the parent of $Y$ (see Figure~\ref{fig:sets}). 
\end{definition}
\begin{figure}[h]
	\begin{center}
		\includegraphics[width=3.7cm]{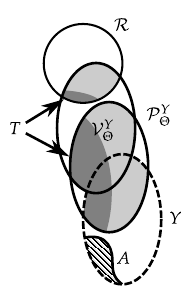}
		\caption{\label{fig:sets} Sets used in Definition~\ref{def:prob-functions}. The set $Y$ is the ellipse with dashed boundary. $\mathcal{P}_{\Theta}^Y$ is the set of bags in the path from $Y$ to the root $\calR$. The set $\mathcal{V}_{\Theta}^Y$ contains all areas depicted in  gray (both light and dark gray). The dark gray part of $\mathcal{V}_{\Theta}^Y$ is $T$. The hatched subset of $Y$ is $A$.}
	\end{center}
\end{figure}
The functions introduced in Definition~\ref{def:prob-functions} have a probabilistic interpretation that will be at the basis of our rounding algorithm.
The structure provided by constraints \eqref{eq:prob-measure}-\eqref{eq:prob-consistent} induces probability distributions over subsets of a bag in the decomposition $\Theta$.
For a bag $Y$, the value \eqref{eq:conditional-prob} can be interpreted as a conditional probability given the choice of $T\subseteq \calV_{\Theta}^Y$.
The following proposition summarizes these properties.
\begin{proposition}
\label{prop:useful-x}	
Consider an instance $(G,D,\capac,\dem)$ with $G$ of treewidth $k$ and let $\ell$ be a positive integer. Let $\Theta=(\calY,\calE)$ be a $(k,\ell)$-decomposition of the graph $G$ and let $(x,y)$ be a solution satisfying \eqref{eq:cut-indicating}-\eqref{eq:prob-consistent}.
Then, the following holds:	
\begin{enumerate}[label=(\alph*)]
	\item Let $L,I\in \calS_{\Theta}$ such that $L\subseteq I$.
	Then, for every $C\subseteq L$, we have $x(L,C)=\sum_{I'\subseteq I\setminus L}x(I,C\cup I')$.\label{prob-extension}
	\item $\sum_{A\subseteq \calR}f_{x,\Theta}^{\calR}(A)=1$.\label{prob-root}
	\item For every non-root bag $Y\in \calY$ and every $T\subseteq \calV_{\Theta}^Y$, we have $\sum_{A\subseteq Y\setminus \mu(Y)}f_{x,\Theta}^{T,Y}(A)=1$.\label{prob-nonroot}
\end{enumerate}
\end{proposition}

\begin{proof}
We prove part \ref{prob-extension} by induction on the size of $I\setminus L$.
When $I\setminus L=\{u\}$ for some node $u\in V$, since $x$ satisfies condition \eqref{eq:prob-consistent} we have 
\begin{align*}
x(L,C)&=x(L+u,C)+x(L+u,C+u)=x(I,C)+x(I,C+u)=\sum_{I'\subseteq I\setminus L}x(I,C\cup I'),
\end{align*}
where the second equality holds since $I=L+u$.
Now consider any pair $L,I\in \calS_{\Theta}$ with $L\subseteq I$ such that the size of $I\setminus L$ is larger than one, consider any node $v\in I\setminus L$.
Then, by condition \eqref{eq:prob-consistent} and the inductive step we get
\begin{align*}
x(L,C)&=x(L+v,C)+x(L+v,C+v)\\
      &=\sum_{H\subseteq I\setminus (L+v)}x(I,C\cup H)+\sum_{H\subseteq I\setminus (L+v)}x(I,(C+v)\cup H)\\
      &=\sum_{H\subseteq I\setminus (L+v)}x(I,C\cup H)+\sum_{H\subseteq I\setminus (L+v)}x(I,C\cup( H+v))\\
      &=\sum_{I'\subseteq I\setminus L}x(I,C\cup I'),
\end{align*}
which concludes the proof for this part.
Part \ref{prob-root} follows directly since $x$ satisfies \eqref{eq:prob-measure} for $S=\calR$ and therefore 
\[\sum_{A\subseteq \calR}f_{x,\Theta}^{\calR}(A)=\sum_{A\subseteq \calR}x(\calR,A)=1.\]
To show part \ref{prob-nonroot}, by applying part \ref{prob-extension} with $L=\calV_{\Theta}^Y$ and $I=\calV_{\Theta}^Y\cup Y$ we get that 
\[x(\calV_{\Theta}^Y,T)=\sum_{A\subseteq Y\setminus \calV_{\Theta}^Y}x(\calV_{\Theta}^Y\cup Y,T\cup A)=\sum_{A\subseteq Y\setminus \mu(Y)}x(\calV_{\Theta}^Y,T)\cdot f_{x,\Theta}^{T,Y}(A),\]
where we used that $Y\setminus \calV_{\Theta}^Y=Y\setminus \mu(Y)$.
That concludes the proof.
\end{proof}
We first design a randomized algorithm to show the existence of $2$-approximation by rounding an optimal solution of the linear fractional program \eqref{eq:objective}-\eqref{eq:prob-consistent} defined by a $(k,\ell)$-decomposition $\Theta$.
We start by constructing a solution at the root level, and then by conditioning on this assignment we construct a solution for the children, and we continue this propagation process until we recover an integral solution.
Theorem \ref{thm:main-sparsest} is finally obtained by optimizing the running time of our algorithm as a function of $\ell$, and by performing a derandomization to get a deterministic $2$-approximation algorithm.
We provide the detailed randomized algorithm below.

\begin{algorithm}[H]
	\begin{algorithmic}[1]
		\Require{An instance $(G,D,\capac,\dem)$ with $G$ of treewidth $k$ and a positive integer number $\ell$.}
		\Ensure{A cut in the nodes $V$.}
		\State Compute a $(k,\ell)$-decomposition $\Theta=(\calY,\calE)$ of $G$.\label{step1}		
		\State Let $(x,y)$ be an optimal solution of \eqref{eq:objective}-\eqref{eq:prob-consistent}.\label{step2}
		\State Sample a subset $B_{\calR}\subseteq \calR$ according to the probability distribution $f_{x,\Theta}^{\calR}$ and let $H_{\calR}=\emptyset$.
		\For{$\ell=1$ to $\depth(\calE)$}
		\State For every bag $Y$ of level $\ell$ in the tree $\calE$, let $H_Y=H_{\mu(Y)}\cup (B_{\mu(Y)}\cap \calJ_Y)$.
		\State Sample a subset of nodes $B_{Y}\subseteq Y\setminus \mu(Y) $ according to the probability distribution $f_{x,\Theta}^{H_Y,Y}$. 									
		\EndFor
		\State Return $\calB=\bigcup_{Y\in \calY}B_{Y}$.
	\end{algorithmic}
	\caption{Randomized Rounding \label{alg:contraction}}
	\label{alg:propagation}
\end{algorithm}

For a bag $Y\in \calY$, the set $B_{\mu(Y)}\cap \calJ_Y$ is a subset of the adhesion of $Y$, and the set $H_Y$ collects the union of these subsets in the path of $\Theta$ that goes from the root to $Y$.
Then, the set $B_{Y}\subseteq Y\setminus \mu(Y)$ is sampled according to a conditional probability that depends on $H_Y$. 
The output of Algorithm \ref{alg:propagation} is a random subset of nodes in $V$ and 
we denote by $\PP_{x,\Theta}$ the probability measure induced by this random set-valued variable.
The following lemmas summarize some properties of the algorithm.

\begin{lemma} 
\label{lem:first-lemma-sparsest}	
Consider $(G,D,\capac,\dem)$ with $G$ of treewidth $k$ and let $\ell$ be a positive integer. Let $\Theta=(\calY,\calE)$ be a $(k,\ell)$-decomposition of $G$ and let $(x,y)$ be a solution satisfying \eqref{eq:cut-indicating}-\eqref{eq:prob-consistent}.
Then, the following holds:	
 \begin{enumerate}[label=(\alph*)]
 	\item For every $Y\in \calY$ and every $S\subseteq Y\cup \calV_{\Theta}^Y$, we have $\PP_{x,\Theta}(\calB\cap S=T)=x(S,T)$ for every $T\subseteq S$.\label{part-a} 
 	\item For every edge $e\in E_G$ in the supply graph, we have $\PP_{x,\Theta}(|e\cap \calB|=1)=y(e)$. \label{part-b}
 \end{enumerate}
\end{lemma} 

\begin{proof}
We prove part \ref{part-a} by induction on the level of the bag $Y$.
Suppose first that $Y=\calR$, that is the root bag of $\Theta$.
In this case we have, by construction, that $\calV_{\Theta}^{\calR}=\emptyset$.
Given $S\subseteq \calR$, we have that $\PP_{x,\Theta}(B_{\calR}=S)=f_{x,\Theta}^{\calR}(S)=x(\calR,S)$.
On the other hand, for any $T\subseteq S$, we have
\begin{align*}
\PP_{x,\Theta}(\calB\cap S=T)&=\PP_{x,\Theta}(B_{\calR}\cap S=T)\\
&=\sum_{T'\subseteq \calR\setminus S}\PP_{x,\Theta}(B_{\calR}=T\cup T')=\sum_{T'\subseteq \calR\setminus S}x(\calR,T\cup T')=x(S,T),
\end{align*}
where the last equality holds by Proposition \ref{prop:useful-x} \ref{prob-extension}.  
Now consider any non-root bag $Y\in \calY$ and let $Z$ be its parent bag.
In particular, we have that $\calV_{\Theta}^Y\subseteq \calV_{\Theta}^{Z}\cup Z$.
Let $S\subseteq \calV_{\Theta}^Y\cup Y$ and consider $S_Y=S\cap \calV_{\Theta}^Y$ and $\tilde S=S\setminus \calV_{\Theta}^Y$.
Then, we have
\begin{align*}
\PP_{x,\Theta}(B_{Y}\cup H_Y=S)&=\PP_{x,\Theta}(H_Y=S_Y,B_{Y}=\tilde S)\\
								&=\PP_{x,\Theta}(H_Y=S_Y)\cdot \PP_{x,\Theta}(B_{Y}=\tilde S|H_Y=S_Y)\\
								&=\PP_{x,\Theta}(\calB\cap \calV_{\Theta}^Y=S_Y)\cdot \PP_{x,\Theta}(B_{Y}=\tilde S|H_Y=S_Y)\\
								&=x(\calV_{\Theta}^Y,S_Y)\cdot f_{x,\Theta}^{S_Y,Y}(\tilde S)\\
								&=x(\calV_{\Theta}^Y,S_Y)\cdot\frac{x(\calV_{\Theta}^Y\cup Y,S_Y\cup \tilde S)}{x(\calV_{\Theta}^Y,S_Y)}\\
								&=x(\calV_{\Theta}^Y\cup Y,S_Y\cup \tilde S)=x(\calV_{\Theta}^Y\cup Y,S),
\end{align*}
where the fourth equality holds since $\calV_{\Theta}^Y\subseteq \calV_{\Theta}^{Z}\cup Z$ together with the inductive step and the choice of $B_Y$ in Algorithm \ref{alg:propagation}. 
Given $S\subseteq \calV_{\Theta}^Y\cup Y$ and $T\subseteq S$, we have
\begin{align*}
	\PP_{x,\Theta}(\calB\cap S=T)&=\PP_{x,\Theta}((B_{Y}\cup H_Y)\cap S=T)\\
	&=\sum_{T'\subseteq (\calV_{\Theta}^Y\cup Y)\setminus S}\PP_{x,\Theta}(B_{Y}\cup H_Y=T\cup T')\\
	&=\sum_{T'\subseteq (\calV_{\Theta}^Y\cup Y)\setminus S}x(\calV_{\Theta}^Y\cup Y,T\cup T')=x(S,T),
\end{align*}
where the last equality holds by Proposition \ref{prop:useful-x} \ref{prob-nonroot}. 
We now prove part \ref{part-b}.
Consider any edge $e=\{u,v\}\in E_G$ and let $Y\in \calY$ be such that $e\in Y$. 
We know this bag exists since $\Theta$ is a tree decomposition by Lemma \ref{lem:super-node}.
Since $x$ satisfies \eqref{eq:cut-indicating}, together with part \ref{part-a} we have that $\PP_{x,\Theta}(|e\cap \calB|=1)=\PP_{x,\Theta}(e\cap \calB=\{u\})+\PP_{x,\Theta}(e\cap \calB=\{v\})=x(e,\{u\})+x(e,\{v\})=y(e)$.
\end{proof}

\begin{lemma}
\label{lem:second-lemma-sparsest}
Consider $(G,D,\capac,\dem)$ with $G$ of treewidth $k$ and let $\ell$ be a positive integer. Let $\Theta=(\calY,\calE)$ be a $(k,\ell)$-decomposition of $G$ and let $(x,y)$ be a solution satisfying \eqref{eq:cut-indicating}-\eqref{eq:prob-consistent}.
Then, for every edge $e\in E_D$ in the demand graph we have $\PP_{x,\Theta}(|e\cap \calB|=1)\ge y(e)/2$.
\end{lemma}

\begin{proof}
Let $e=\{s,t\}\in E_D$ be a demand edge.
When $e\in E_G$ we are done since $\PP_{x,\Theta}(|e\cap \calB|=1)= y(e)$ by Lemma \ref{lem:first-lemma-sparsest} \ref{part-b}.
Suppose in what follows that $e\notin E_G$, and let $Y_s$ and $Y_t$ be the least depth bags in the tree $\calE$ such that $s\in Y_s$ and $t\in Y_t$.
Furthermore, let $Y$ be the lowest common ancestor of the bags $Y_s, Y_t$ in the tree $\calY$. 
Let $\calC_{e}=(Y_s\cup \calV_{\Theta}^{Y_s})\cup (Y_t\cup\calV_{\Theta}^{Y_t})$.
For every $T\subseteq \calC_{e}$ consider the value $g_e(T)=x(\calC_e,T)$.
Since $x$ satisfies \eqref{eq:prob-measure}, we have that 
$\sum_{T\subseteq \calC_{e}}g_e(T)=1$ and therefore $g_e$ defines a probability mass function over $\calS_{\Theta}(Y_s,Y_t)$, which is the power set of $\calC_e$.
Consider the set-valued random variable $W$ distributed according to $g_e$ and let $\QQ_e$ the probability measure induced by this random variable.
Then, we have
\begin{align*}
\QQ_e(|e\cap W|=1)&=\QQ_e(e\cap W=\{s\})+\QQ_e(e\cap W=\{t\})\\
				  &=\sum_{C'\subseteq \calC_e\setminus \{e\}}x(\calC_e,s+C')+\sum_{C'\subseteq \calC_e\setminus \{e\}}x(\calC_e,t+C')\\
				  &=x(e,s)+x(e,t)=y(e),
\end{align*}
where the third equality holds by Proposition \ref{prop:useful-x} \ref{prob-extension} and the last equality holds since $x$ satisfies condition \eqref{eq:cut-indicating}.
Let $Z_s$ and $Z_t$ be the children bags of $Y$ such that $Z_s$ belongs to unique path from $Y_s$ to the root $\calR$ and $Z_t$ belongs to unique path from $Y_t$ to the root $\calR$, in the tree $\calE$.
Define the set $\Lambda=\calV_{\Theta}^Y\cup \calJ_{Z_s}\cup \calJ_{Z_t}$. 
Observe that 
\begin{align*}
	\PP_{x,\Theta}(|e\cap \calB|=1)&=\sum_{T\subseteq \Lambda}\PP_{x,\Theta}(|e\cap \calB|=1\;|\;\calB\cap \Lambda=T)\cdot \PP_{x,\Theta}(\calB\cap \Lambda=T)\\
	&=\sum_{T\subseteq \Lambda}\PP_{x,\Theta}(|e\cap \calB|=1\;|\;\calB\cap \Lambda=T)\cdot x(\Lambda,T),
\end{align*}
where the last equality holds by Lemma \ref{lem:first-lemma-sparsest} \ref{part-a} and the fact that $\Lambda\subseteq \calV_{\Theta}^Y\cup Y$.
On the other hand, for any $L\subseteq W$ and every $I\subseteq L$ we have
\begin{equation}
\QQ_{e}(W\cap L=I)=\sum_{C'\subseteq \calC_e\setminus \Lambda}x(\calC_e,I\cup C')=x(L,T),\label{eq:thing1}
\end{equation}
where the last equality holds by Proposition \ref{prop:useful-x} \ref{prob-extension}. 
Therefore, we have 
\begin{align*}
	y(e)=\QQ_{e}(|e\cap W|=1)&=\sum_{T\subseteq \Lambda}\QQ_{e}(|e\cap W|=1\;|\;W\cap \Lambda=T)\cdot \QQ_{e}(W\cap \Lambda=T)\\
	&=\sum_{T\subseteq \Lambda}\QQ_{e}(|e\cap W|=1\;|\;W\cap \Lambda=T)\cdot x(\Lambda,T), 
\end{align*}
where the last equality holds by applying \eqref{eq:thing1} with $L=\Lambda$.
Then, in order to conclude the lemma it is sufficient to show that $\QQ_{e}(|e\cap W|=1\;|\;W\cap \Lambda=T)\le 2\cdot \PP_{x,\Theta}(|e\cap \calB|=1\;|\;\calB\cap \Lambda=T)$.
Given $T\subseteq \Lambda$, consider the random variable $\omega_{s,T}\in \{0,1\}$ that indicates whether $s\in W$ given $W\cap \Lambda=T$, and let $\beta_{s,T}\in \{0,1\}$ be the random variable that indicates whether $s\in \calB$ given $\calB\cap \Lambda=T$.
We define analogously the random variables $\omega_{t,T}$ and $\beta_{t,T}$.
Since $s,t\notin \Lambda$, we observe that for any $T\subseteq \Lambda$ and $v\in \{s,t\}$ it holds that $v\in W$ and $W\cap \Lambda=T$ if and only if $W\cap (\Lambda+v)=T+v$.
Therefore, for every $T\subseteq \Lambda$, we have that 
\begin{align*}
\QQ_e(\omega_{s,T}=1)&=\frac{\QQ_e(s\in W,W\cap \Lambda=T)}{\QQ_e(W\cap \Lambda=T)}=\frac{x(\Lambda+s,T+s)}{x(\Lambda,T)}=\PP_{x,\Theta}(\beta_{s,T}=1),\\
\QQ_e(\omega_{t,T}=1)&=\frac{\QQ_e(t\in W,W\cap \Lambda=T)}{\QQ_e(W\cap \Lambda=T)}=\frac{x(\Lambda+t,T+t)}{x(\Lambda,T)}=\PP_{x,\Theta}(\beta_{t,T}=1),
\end{align*}
where, in both cases, the first equality comes from the above observation and \eqref{eq:thing1} and the second equality is a consequence of the above observation and Proposition \ref{prop:useful-x} \ref{prob-extension}.
We conclude that for every $T\subseteq \Lambda$ the random variables $\omega_{v,T}$ and $\beta_{v,T}$ are identically distributed, for $v\in \{s,t\}$.
Furthermore, by construction in Algorithm \ref{alg:propagation}, the random variables $\omega_{s,T}$ and $\beta_{s,T}$ are independent.
\begin{claim}
	Suppose we have two random variables $G$ and $K$ taking values in $\{0,1\}$.
	Then, we have that $\Pr(G\ne K)\le 2(\Pr(G=1)\Pr(K=0)+ \Pr(G=0)\Pr(K=1))$.
\end{claim}
\noindent We show how to conclude the lemma using the claim. 
Taking $G=\omega_{s,T}$ and $K=\omega_{t,T}$, we have
\begin{align*}
	\QQ_e(|e\cap W|=1\;|\;W\cap \Lambda=T)&=\QQ_e(\omega_{s,T}\ne \omega_{t,T})\\
	&\le 2(\QQ_e(\omega_{s,T}=1)\QQ_e(\omega_{t,T}=0)+ \QQ_e(\omega_{s,T}=0)\QQ_e(\omega_{t,T}=1))\\
	&= 2(\PP_{x,\Theta}(\beta_{s,T}=1)\PP_{x,\Theta}(\beta_{t,T}=0)+ \PP_{x,\Theta}(\beta_{s,T}=0)\PP_{x,\Theta}(\beta_{t,T}=1))\\
	&=2\cdot \PP_{x,\Theta}(\beta_{s,T}\ne \beta_{t,T})=2\cdot \PP_{x,\Theta}(|e\cap \calB|=1\;|\;\calB\cap \Lambda=T),
\end{align*}
which concludes the lemma.
We now show how to prove the claim.
Let $a=\Pr(G=1)$ and $b=\Pr(K=1)$.
Then, we have that $\Pr(G\ne K)\le a(1-b)+b(1-a)+\min\{ab,(1-a)(1-b)\}$.
We prove next that $\min\{ab,(1-a)(1-b)\}\le a(1-b)+b(1-a)$ for every $a,b\in [0,1]$.
Having this inequality, we get that $\Pr(G\ne K)\le 2a(1-b)+2b(1-a)$ and therefore the claim holds.
To prove the inequality consider two cases.
Suppose first that $a+b\le 1$.
In particular, we have that $ab\le 1-a-b+ab=(1-a)(1-b)$.
Furthermore, $ab\le a^2+b^2\le a(1-b)+b(1-a)$, where the last inequality holds since $a,b\in [0,1]$ and $a+b\le 1$ implies that $a\le 1-b$ and $b\le 1-a$.
Now suppose that $a+b>1$.
In particular, we have $ab>1-a-b+ab=(1-a)(1-b)$.
On the other hand, $(1-a)(1-b)\le (1-a)^2+(1-b)^2<b(1-a)+a(1-b)$, where the last inequality holds since $a,b\in [0,1]$ and $a+b>1$ implies that $a>1-b$ and $b>1-a$.
This concludes the proof of the inequality.
\end{proof}

\begin{definition}
Let $G$ be a graph of treewidth $k$, let $\Theta=(\calY,\calE)$ be a $(k,\ell)$-decomposition of $G$, and consider a node $u\in V$.
Let $X$ be the least depth bag in the tree containing the node $u$.
Given a bag $Z\in \calP_{\Theta}^X$ and $H\subseteq \calV_{\Theta}^Z\cup Z$, we say that a pair $(M,N)$ is an $H$-extension for the node $u$ if the following holds:
\begin{enumerate}[label=(\roman*)]
 	\item $N\subseteq X\setminus \calV_{\Theta}^X$ and $u\in N$,
 	\item $M=(H\cap \calV_{\Theta}^X)\cup L$ where $L\subseteq \calV_{\Theta}^X\setminus (\calV_{\Theta}^Z\cup Z)$.
\end{enumerate}	
We denote by $\Delta_{\Theta}(H,u)$ the set of $H$-extensions for $u$.
\end{definition}
Observe that for any node $u$ and $X$ being the least depth bag containing $u$, for any bag $Z\in \calP_{\Theta}^X$ and any $H\subseteq \calV_{\Theta}^Z\cup Z$, the set $\Delta_{\Theta}(H,u)$ has cardinality at most 
\begin{equation}2^{O(k(2^{\ell}+\log(n)/\ell))}\label{obs:delta-size}\end{equation}
when $\Theta$ is a $(k,\ell)$-decomposition.
This holds since, by Lemma \ref{lem:super-node}, we have $|X\setminus \calV_{\Theta}^X|\in O(2^{\ell}k)$ and $|\calV_{\Theta}^X\setminus (\calV_{\Theta}^Z\cup Z)|\in O(k\log(n)/\ell)$.
We need one more lemma before proving Theorem \ref{thm:main-sparsest}.

\begin{lemma}
\label{lem:lambert}
For every positive real value $x\ge 4$, there exists a unique value $\alpha^{\star}$ such that alpha $\alpha^{\star}2^{\alpha^{\star}} = x$, and it satisfies the inequality
$2^{\lceil \alpha^{\star}\rceil}+x/\lceil \alpha^{\star}\rceil\le 12x/\log(x)$.
\end{lemma}

\begin{proof}
It is well-known that the unique real solution of the equation $\alpha 2^\alpha=x$ is given by $W_0(x\ln(2))/\ln(2)$, where $W_0$ is the principal branch of the Lambert function \cite{corless1996lambertw}.
Furthermore, we have $W_0(z)\ge \ln(z)-\ln\ln(z)\ge \ln(z)/2$ for every $z\ge e$, and therefore we get 
\begin{equation}
\alpha^{\star}=\frac{W_0(x\ln(2))}{\ln(2)}\ge \frac{\ln(x)+\ln\ln(2)}{2\ln(2)}\ge 0.4\ln(x)\ge 0.27\log(x)\label{ineq:lambert}
\end{equation} 
for every $x\ge 4$.
Then, we have 
\begin{equation*}
2^{\lceil \alpha^{\star}\rceil}+\frac{x}{\lceil \alpha^{\star}\rceil}\le 2^{ \alpha^{\star}+1}+\frac{x}{ \alpha^{\star}}=\frac{3x}{ \alpha^{\star}}< \frac{12x}{\log(x)},
\end{equation*}
where the equality holds since $\alpha^{\star}$ solves $\alpha 2^\alpha=x$ and the last inequality holds by \eqref{ineq:lambert}.
\end{proof}

\begin{proof}[Proof of Theorem \ref{thm:main-sparsest}]
Let $\calI=(G,D,c,d)$ be an instance of the non-uniform sparsest cut problem.
Recall that we denote by $n$ the number of nodes in the instance.
Let $\alpha^{\star}_n$ be the unique positive real solution of the equation $\alpha 2^\alpha =\log(n)$ and let $\ell^{\star}=\lceil \alpha^{\star}_n\rceil$.
We run Algorithm \ref{alg:propagation} over the instance $\calI$, using the value $\ell^{\star}$, and let $\Theta$ be the $(k,\ell^{\star})$-decomposition computed in step 1 of the algorithm. 
Let $(x,y)$ be an optimal solution of the optimization problem \eqref{eq:objective}-\eqref{eq:prob-consistent} solved in step 2 of the algorithm, and we denote by $\opt_{\LP}$ the optimal value $\sum_{e\in E_G}\capac(e)y(e)/\sum_{e\in E_D}\dem(e)y(e)$. 
Let $\calB$ be the solution computed by the randomized algorithm.
For every pair of nodes $e=\{u,v\}\subseteq V$, with $u\ne v$, let $\xi(e)$ be equal to one if $|e\cap \calB|=1$ and zero otherwise.
This random variable indicates when a pair of nodes is cut by the algorithm solution.
Consider $\calC=\sum_{e\in E_G}\capac(e)\xi(e)$ and $\calD=\sum_{e\in E_D}\dem(e)\xi(e)$.
By Lemmas \ref{lem:first-lemma-sparsest} \ref{part-b} and \ref{lem:second-lemma-sparsest} we have that 
\begin{align*}
\EE_{x,\Theta}(\calC)&=\sum_{e\in E_G}\capac(e)\cdot \EE_{x,\Theta}(\xi(e))=\sum_{e\in E_G}\capac(e)\cdot \PP_{x,\Theta}(|e\cap \calB|=1)=\sum_{e\in E_G}\capac(e)y(e),\\
\EE_{x,\Theta}(\calD)&=\sum_{e\in E_D}\capac(e)\cdot \EE_{x,\Theta}(\xi(e))=\sum_{e\in E_D}\dem(e)\cdot \PP_{x,\Theta}(|e\cap \calB|=1)\ge \frac{1}{2}\sum_{e\in E_D}\dem(e)y(e),
\end{align*}
and therefore we get $\EE_{x,\Theta}(\calC)/\EE_{x,\Theta}(\calD)\le 2\cdot \opt_{\LP}\le 2\cdot\min_{S\subseteq V} \phi(S)$, where the last inequality holds since the sparsest cut of value $\min_{S\subseteq V} \phi(S)$ defines a feasible solution for \eqref{eq:objective}-\eqref{eq:prob-consistent}.
We now show how to derandomize the solution $\calB$ to get a deterministic 2-approximation.
We use the method of conditional expectations.
Define the random variable $\Gamma=\calC-2\cdot \calD\cdot \opt_{\LP}$.
Then, we have that $0\ge \EE_{x,\Theta}(\Gamma)=\EE(\EE_{x,\Theta}(\Gamma|B_{\calR}))$ 
and therefore there exists $R' \subseteq \calR$ such that $\EE_{x,\Theta}(\Gamma|B_{\calR}=R')\le 0$.
Fix any subset $Y'_1\subseteq \calR$ with $\EE_{x,\Theta}(\Gamma|B_{\calR}=Y'_1)\le 0$ and let $\calR=Y_1,Y_2,\ldots,Y_{|\calY|}$ be the bags visited according to some BFS ordering. 
Suppose we have computed for some $t\in \{1,\ldots,|\calY|-1\}$ the set $A_t=\cup_{\ell=1}^t Y'_{\ell}\subseteq \cup_{\ell=1}^tY_{\ell}$, with $Y'_{\ell}\subseteq Y_{\ell}\setminus \mu(Y_{\ell})$ for each $\ell\in \{1,\ldots,t\}$, and such that $\EE_{x,\Theta}(\Gamma|\calB\cap (\cup_{\ell=1}^tY_{\ell})=A_t)\le 0$.
Then, we have 
\begin{align*}
0&\ge \EE_{x,\Theta}(\Gamma|\calB\cap (\cup_{\ell=1}^tY_{\ell})=A_t)\\
&=\sum_{Y'\subseteq Y\setminus \mu(Y)}\EE_{x,\Theta}(\Gamma|\calB\cap (\cup_{\ell=1}^tY_{\ell})=A_t, B_{Y_{t+1}}=Y')\cdot \PP_{x,\Theta}(B_{Y_{t+1}}=Y')\\
&=\sum_{Y'\subseteq Y\setminus \mu(Y)}\EE_{x,\Theta}(\Gamma|\calB\cap (\cup_{\ell=1}^{t+1}Y_{\ell})=A_t\cup Y')\cdot \PP_{x,\Theta}(B_{Y_{t+1}}=Y'),
\end{align*}
and therefore there exists $Y'\subseteq Y_{t+1}\setminus \mu(Y_{t+1})$ such that $\EE_{x,\Theta}(\Gamma|\calB\cap (\cup_{\ell=1}^{t+1}Y_{\ell})=A_t\cup Y')\le 0$.
Fix any of these subsets and we denote it by $Y'_{t+1}$. 
By the end of this process, let $\calA$ be the union of $Y'_1,\ldots,Y'_{|\calY|}$.
By construction, we have recovered a solution such that $\EE_{x,\Theta}(\Gamma|\calB=\calA)\le 0$ and therefore $\calA$ is a 2-approximation.

We now study the running time of the derandomization, and more specifically, the running time that we need to compute the conditional expectations.
Let $t\in \{1,\ldots,|\calY|\}$ and let $T\subseteq \cup_{\ell=1}^tY_{\ell}$.
To compute the value of the expectation $\EE_{x,\Theta}(\Gamma|\calB\cap (\cup_{\ell=1}^tY_{\ell})=T)$, it is sufficient to compute the probability value $\PP_{x,\Theta}(|e\cap \calB|=1|\calB\cap (\cup_{\ell=1}^tY_{\ell})=T)$ for any $e\in E_G$ or $e\in E_D$.
Furthermore, when $e\subseteq \cup_{\ell=1}^tY_{\ell}$ the value of the probability is determined and equal to one or zero.
Then, we suppose that $e=\{u,v\}$ is not contained in $\cup_{\ell=1}^tY_{\ell}$.
For every node $a\in V\setminus (Y_1\cup \cdots \cup Y_t)$ let $X_a$ be the least depth bag in $\calY$ that contains $a$.
In particular, we have that $X_a\notin \{Y_1,\ldots,Y_t\}$ and let $Z_a$ be the lowest bag in $\{Y_1,\ldots,Y_t\}$ such that $Z_a$ belongs to the path from $X_a$ to the root.
For every $a\in V\setminus (Y_1\cup \cdots \cup Y_t)$ consider the quantity
\[g_a=\PP_{x,\Theta}\Big(a\in \calB\;\Big|\;\calB\cap (\calV_{\Theta}^{Z_a}\cup Z_a)=T\cap (\calV_{\Theta}^{Z_a}\cup Z_a)\Big).\]
\noindent{\bf Case 1.} Suppose that $u\notin \cup_{\ell=1}^tY_{\ell}$ and $v\notin \cup_{\ell=1}^tY_{\ell}$ and that $Z_u\ne Z_v$.
Then, $\calP_{\Theta}^{Z_u}$ and $\calP_{\Theta}^{Z_v}$ are contained in the subtree induced by the bags $Y_1,\ldots,Y_t$.
By construction in Algorithm \ref{alg:propagation} we have that $\PP_{x,\Theta}(|e\cap \calB|=1|\calB\cap (\cup_{\ell=1}^tY_{\ell})=T)=g_u(1-g_v)+g_v(1-g_u)$.
Furthermore, by denoting $T_a=T\cap (\calV_{\Theta}^{Z_a}\cup Z_a)$, we have
\[g_a=\sum_{(M,N)\in \Delta_{\Theta}(T_a,u)}f_{x,\Theta}^{M,X_a}(N)=\sum_{(M,N)\in \Delta_{\Theta}(T_a,u)}\frac{x(\calV_{\Theta}^{X_a}\cup X_a,M\cup N)}{x(\calV_{\Theta}^{X_a},M)}\]
for each $a\in \{u,v\}$.
By the observation in \eqref{obs:delta-size}, $g_u$ and $g_v$ can be computed in time $2^{O(k(2^{\ell^{\star}}+\log(n)/\ell^{\star}))}$.\\

\noindent{\bf Case 2.} Suppose that $u\notin \cup_{\ell=1}^tY_{\ell}$ and $v\notin \cup_{\ell=1}^tY_{\ell}$, and that $Z_u=Z_v=Z$.
Let $W$ be the lowest common ancestor of $Y_u$ and $Y_v$.
In particular, $Z$ is an ancestor of $W$ and $W\notin \{Y_1,\ldots,Y_t\}$.
For every $H\subseteq \calV_{\Theta}^W\setminus (\calV_{\Theta}^Z\cup Z)$ and $K\subseteq W\setminus \mu(W)$ consider the quantity
\[\beta(H,K)=\frac{x(\calV_{\Theta}^{W}\cup W,(T\cap\calV_{\Theta}^W)\cup H\cup K)}{x(\calV_{\Theta}^{W},(T\cap\calV_{\Theta}^W)\cup H)}.\]
Furthermore, for every $H\subseteq \calV_{\Theta}^W\setminus (\calV_{\Theta}^Z\cup Z)$, $K\subseteq W\setminus \mu(W)$ and $a\in \{u,v\}$ let 
\[\gamma_a(H,K)=\sum_{(M,N)\in \Delta_{\Theta}((T\cap\calV_{\Theta}^W)\cup H\cup K,u)}\frac{x(\calV_{\Theta}^{Y_a}\cup Y_a,M\cup N)}{x(\calV_{\Theta}^{Y_a},M)}\]
Then, we have that $\PP_{x,\Theta}(|e\cap \calB|=1|\calB\cap (\cup_{\ell=1}^tY_{\ell})=T)$ is equal to 
\[\sum_{H\subseteq \calV_{\Theta}^W\setminus (\calV_{\Theta}^Z\cup Z)}\sum_{K\subseteq W\setminus \mu(W)}\beta(H,K)\Big(\gamma_u(H,K)(1-\gamma_v(H,K))+\gamma_v(H,K)(1-\gamma_u(H,K))\Big).\]
As before, the above summation can be computed in time $2^{O(k(2^{\ell^{\star}}+\log(n)/\ell^{\star}))}$.\\

\noindent{\bf Case 3.} Suppose that $u\in \cup_{\ell=1}^tY_{\ell}$ and $v\notin \cup_{\ell=1}^tY_{\ell}$ (the other case is symmetric).
In this case, we have that $\PP_{x,\Theta}(|e\cap \calB|=1|\calB\cap (\cup_{\ell=1}^tY_{\ell})=T)$ is equal to $1-g_v$, and therefore we can compute it in time \[2^{O(k(2^{\ell^{\star}}+\log(n)/\ell^{\star}))}.\]

As we observe at the end of Section \ref{sec:sparsest-relaxation}, the optimization problem \eqref{eq:objective}-\eqref{eq:prob-consistent} can be solved in time
\[2^{O(k(2^{\ell^{\star}}+\log(n)/\ell^{\star}))}|\calI|^{O(1)}.\]
On the other hand, for every $n\ge 16$, by Lemma \ref{lem:lambert} we have 
\[k2^{\ell^{\star}}+\frac{k\log(n)}{\ell^{\star}}=k2^{\lceil \alpha^{\star}_n\rceil}+\frac{k\log(n)}{\lceil \alpha^{\star}_n\rceil}\le \frac{12k\log(n)}{\log\log(n)},\] 
and therefore, the randomized algorithm and the derandomization can be all performed in time 
\[2^{O\Big(k\frac{\log(n)}{\log\log(n)}\Big)} |\calI|^{O(1)}=2^{2^{O(k)}}|\calI|^{O(1)}.\]
To finish the proof, we verify the above equality by considering two cases. 
If $k<\log\log(n)$, we have $k\log(n)/\log\log(n)<\log(n)$ and the equality holds.
Otherwise, if $k\ge \log\log(n)$ and $n\ge 4$ we have 
$k\log(n)/\log\log(n)\le k\log(n)=2^{\log(k)+\log\log(n)}\le 2^{\log(k)+k}=2^{O(k)}$.
\end{proof}

\bibliographystyle{abbrv}
{\small \bibliography{refs}}

\begin{thebibliography}{10}

\bibitem{aprile2020tight}
M.~Aprile, M.~Drescher, S.~Fiorini, and T.~Huynh.
\newblock A tight approximation algorithm for the cluster vertex deletion
  problem.
\newblock In {\em Integer Programming and Combinatorial Optimization (IPCO)},
  2021.

\bibitem{arora2008euclidean}
S.~Arora, J.~Lee, and A.~Naor.
\newblock Euclidean distortion and the sparsest cut.
\newblock {\em Journal of the American Mathematical Society}, 21(1):1--21,
  2008.

\bibitem{DBLP:journals/jacm/AroraRV09}
S.~Arora, S.~Rao, and U.~V. Vazirani.
\newblock Expander flows, geometric embeddings and graph partitioning.
\newblock {\em Journal of the {ACM}}, 56(2):5:1--5:37.

\bibitem{bienstock2018lp}
D.~Bienstock and G.~Munoz.
\newblock Lp formulations for polynomial optimization problems.
\newblock {\em SIAM Journal on Optimization}, 28(2):1121--1150, 2018.

\bibitem{bienstock2004tree}
D.~Bienstock and N.~Ozbay.
\newblock Tree-width and the sherali--adams operator.
\newblock {\em Discrete Optimization}, 1(1):13--21, 2004.

\bibitem{bodlaender1988nc}
H.~L. Bodlaender.
\newblock Nc-algorithms for graphs with small treewidth.
\newblock In {\em International Workshop on Graph-Theoretic Concepts in
  Computer Science (WG)}, 1988.

\bibitem{boyd2004convex}
S.~Boyd and L.~Vandenberghe.
\newblock {\em Convex optimization}.
\newblock Cambridge university press, 2004.

\bibitem{chakrabarti2008embeddings}
A.~Chakrabarti, A.~Jaffe, J.~R. Lee, and J.~Vincent.
\newblock Embeddings of topological graphs: Lossy invariants, linearization,
  and 2-sums.
\newblock In {\em IEEE Symposium on Foundations of Computer Science (FOCS)},
  2008.

\bibitem{CKMSUV21}
P.~Chalermsook, M.~Kaul, M.~Mnich, J.~Spoerhase, S.~Uniyal, and D.~Vaz.
\newblock Approximating sparsest cut in low-treewidth graphs via combinatorial
  diameter.
\newblock {\em CoRR}, 2021.

\bibitem{chawla2006hardness}
S.~Chawla, R.~Krauthgamer, R.~Kumar, Y.~Rabani, and D.~Sivakumar.
\newblock On the hardness of approximating multicut and sparsest-cut.
\newblock {\em computational complexity}, 15(2):94--114, 2006.

\bibitem{chekuri2013flow}
C.~Chekuri, F.~B. Shepherd, and C.~Weibel.
\newblock Flow-cut gaps for integer and fractional multiflows.
\newblock {\em Journal of Combinatorial Theory, Series B}, 103(2):248--273,
  2013.

\bibitem{chlamtac2010approximating}
E.~Chlamtac, R.~Krauthgamer, and P.~Raghavendra.
\newblock Approximating sparsest cut in graphs of bounded treewidth.
\newblock In {\em Approximation, Randomization, and Combinatorial Optimization.
  Algorithms and Techniques (APPROX/RANDOM)}, 2010.

\bibitem{cohen2021quasipolynomial}
V.~Cohen-Addad, A.~Gupta, P.~N. Klein, and J.~Li.
\newblock A quasipolynomial $(2+\varepsilon)$-approximation for planar sparsest
  cut.
\newblock In {\em ACM Symposium on Theory of Computing (STOC)}, 2021.

\bibitem{corless1996lambertw}
R.~M. Corless, G.~H. Gonnet, D.~E. Hare, D.~J. Jeffrey, and D.~E. Knuth.
\newblock On the lambertw function.
\newblock {\em Advances in Computational mathematics}, 5(1):329--359, 1996.

\bibitem{davies2021scheduling}
S.~Davies, J.~Kulkarni, T.~Rothvoss, J.~Tarnawski, and Y.~Zhang.
\newblock Scheduling with communication delays via lp hierarchies and
  clustering ii: Weighted completion times on related machines.
\newblock In {\em ACM-SIAM Symposium on Discrete Algorithms (SODA)}, 2021.

\bibitem{Garg18}
S.~Garg.
\newblock Quasi-ptas for scheduling with precedences using {LP} hierarchies.
\newblock In {\em International Colloquium on Automata, Languages, and
  Programming (ICALP)}, 2018.

\bibitem{gupta2004cuts}
A.~Gupta, I.~Newman, Y.~Rabinovich, and A.~Sinclair.
\newblock Cuts, trees and l1-embeddings of graphs.
\newblock {\em Combinatorica}, 24(2):233--269, 2004.

\bibitem{GTW13}
A.~Gupta, K.~Talwar, and D.~Witmer.
\newblock Sparsest cut on bounded treewidth graphs: algorithms and hardness
  results.
\newblock In {\em ACM Symposium on Theory of Computing (STOC)}, 2013.

\bibitem{khot2015unique}
S.~A. Khot and N.~K. Vishnoi.
\newblock The unique games conjecture, integrality gap for cut problems and
  embeddability of negative-type metrics into l1.
\newblock {\em Journal of the ACM}, 62(1):1--39, 2015.

\bibitem{klein1990approximation}
P.~Klein, A.~Agrawal, R.~Ravi, and S.~Rao.
\newblock Approximation through multicommodity flow.
\newblock In {\em IEEE Symposium on Foundations of Computer Science (FOCS)},
  1990.

\bibitem{klein1995approximate}
P.~Klein, S.~Rao, A.~Agrawal, and R.~Ravi.
\newblock An approximate max-flow min-cut relation for undirected
  multicommodity flow, with applications.
\newblock {\em Combinatorica}, 15(2):187--202, 1995.

\bibitem{lee2010coarse}
J.~R. Lee and P.~Raghavendra.
\newblock Coarse differentiation and multi-flows in planar graphs.
\newblock {\em Discrete \& Computational Geometry}, 43(2):346--362, 2010.

\bibitem{lee2013pathwidth}
J.~R. Lee and A.~Sidiropoulos.
\newblock Pathwidth, trees, and random embeddings.
\newblock {\em Combinatorica}, 33(3):349--374, 2013.

\bibitem{magen2009}
A.~Magen and M.~Moharrami.
\newblock Robust algorithms for on minor-free graphs based on the sherali-adams
  hierarchy.
\newblock In {\em Approximation, Randomization, and Combinatorial Optimization.
  Algorithms and Techniques (APPROX/RANDOM)}, 2009.

\bibitem{maiti2020scheduling}
B.~Maiti, R.~Rajaraman, D.~Stalfa, Z.~Svitkina, and A.~Vijayaraghavan.
\newblock Scheduling precedence-constrained jobs on related machines with
  communication delay.
\newblock In {\em IEEE Symposium on Foundations of Computer Science (FOCS)},
  2020.

\bibitem{matula1990sparsest}
D.~W. Matula and F.~Shahrokhi.
\newblock Sparsest cuts and bottlenecks in graphs.
\newblock {\em Discrete Applied Mathematics}, 27(1-2):113--123, 1990.

\bibitem{okamura1981multicommodity}
H.~Okamura and P.~D. Seymour.
\newblock Multicommodity flows in planar graphs.
\newblock {\em Journal of Combinatorial Theory, Series B}, 31(1):75--81, 1981.

\bibitem{SA90}
H.~D. Sherali and W.~P. Adams.
\newblock A hierarchy of relaxations between the continuous and convex hull
  representations for zero-one programming problems.
\newblock {\em SIAM Journal on Discrete Mathematics}, 3(3):411--430, 1990.

\bibitem{VVW20}
V.~Verdugo, J.~Verschae, and A.~Wiese.
\newblock Breaking symmetries to rescue sum of squares in the case of makespan
  scheduling.
\newblock {\em Mathematical Programming}, 183:583--618, 2020.

\end{thebibliography}

\end{document}